\documentclass[11pt]{article}
\usepackage{amsmath,amsthm}
\usepackage{amssymb,latexsym}
\usepackage{hyperref}
\usepackage{mathpazo}
\usepackage{fullpage}
\usepackage{enumerate}
\usepackage{times}

\renewcommand{\paragraph}[1]{\vspace{1.5mm}\noindent \textbf{#1}}

\mathchardef\mhyphen="2D %hyphen used in math mode

\newtheorem{theorem}{Theorem}

\newtheorem{corollary}[theorem]{Corollary}
\newtheorem{lemma}[theorem]{Lemma}

\newtheorem{definition}[theorem]{Definition}
\newtheorem{claim}[theorem]{Claim}

\newenvironment{proof-sketch}{\noindent{\bf Sketch of Proof:}\hspace*{1em}}{\qed\bigskip}
\newenvironment{proof-idea}{\noindent{\bf Proof Idea:}\hspace*{1em}}{\qed\bigskip}
\newenvironment{proof-of-lemma}[1]{\noindent{\bf Proof of Lemma #1:}\hspace*{1em}}{\qed\bigskip}
\newenvironment{proof-of-proposition}[1]{\noindent{\bf Proof of Proposition #1:}\hspace*{1em}}{\qed\bigskip}
\newenvironment{proof-of-theorem}[1]{\noindent{\bf Proof of Theorem #1:}\hspace*{1em}}{\qed\bigskip}
\newenvironment{proof-attempt}{\noindent{\bf Proof Attempt:}\hspace*{1em}}{\qed\bigskip}

\newcommand{\sig}{{\sigma}}
\newcommand{\eps}{\varepsilon}
\newcommand{\logeps}{\log\eps^{-1}}
\newcommand{\E}{\mathbb{E}}
\newcommand{\cE}{\mathcal{E}}
\newcommand{\cF}{\mathcal{F}}

\newcommand{\tr}[1]{Tr(#1)}
\newcommand{\tensor}{\otimes}
\newcommand{\bra}[1]{\langle #1|}
\newcommand{\ket}[1]{|#1\rangle}
\newcommand{\braket}[2]{\langle #1|#2\rangle}
\newcommand{\ketbra}[2]{\ket{#1}{\bra{#2}}}

\newcommand{\half}{\frac 1 2}
\newcommand{\set}[1]{{\left\{ #1\right\}}}
\newcommand{\zo}{\set{0,1}}
\newcommand{\abs}[1]{\left| #1 \right|}
\newcommand{\norm}[2]{\left| #2 \right|_\mathrm{#1}}
\newcommand{\dnorm}[2]{\left\| #2 \right\|_\mathrm{#1}}
\newcommand{\bignorm}[2]{\big| #2 \big|_\mathrm{#1}}

\newcommand{\remove}[1]{}

 \author{Roy Kasher\footnote{Blavatnik School of Computer Science, Tel Aviv University, Tel Aviv 69978, Israel. Supported by
 JK's Individual Research Grant of the Israeli Science Foundation.}
 \and
 Julia Kempe\footnote{Blavatnik School of Computer Science, Tel Aviv University, Tel Aviv 69978, Israel.
Supported by the European Commission under the Integrated Project Qubit Applications (QAP) funded by the IST
directorate as Contract Number 015848, by an Alon Fellowship of the Israeli Higher Council of Academic Research, by an
Individual Research Grant of the Israeli Science Foundation, by a European Research Council (ERC) Starting Grant, by a
Raymond and Beverly Sackler Career Development Chair and by the Wolfson Family Charitable Trust.}
 }
\title{\bf Two-Source Extractors Secure Against Quantum Adversaries}
\begin{document}
\maketitle

\setcounter{page}{0}\thispagestyle{empty}

\begin{abstract}
We initiate the study of multi-source extractors in the quantum world. In this setting, our goal is to extract random
bits from two independent weak random sources, on which two quantum adversaries store a bounded amount of information.
Our main result is a two-source extractor secure against quantum adversaries, with parameters closely matching the
classical case and tight in several instances. Moreover, the extractor is secure even if the adversaries share
entanglement. The construction is the Chor-Goldreich~\cite{CG88} two-source inner product extractor and its multi-bit
variant by Dodis et al.~\cite{DEOR04}. Previously, research in this area  focused on the construction of seeded
extractors secure against quantum adversaries; the multi-source setting poses new challenges, among which is the
presence of entanglement that could potentially break the independence of the sources.
\end{abstract}

\newpage

\section{Introduction and Results}

Randomness extractors are fundamental in many areas of computer science, with numerous applications to derandomization,
error-correcting codes, expanders, combinatorics and cryptography, to name just a few. Randomness extractors generate
almost uniform randomness from imperfect sources, as they appear either in nature, or in various applications.
Typically, the imperfect source is modelled as a distribution over $n$-bit strings whose {\em min-entropy} is at least
$k$, i.e., a distribution in which no string occurs with probability greater than $2^{-k}$ \cite{SV84,CG88,Zuc90}. Such
sources are known as {\em weak sources}. One way to arrive at a weak source is to imagine that an adversary (or some
process in nature), when in contact with a uniform source, {\em stores} $n-k$ bits of information about the string
(which are later used to break the security of the extractor, i.e. to distinguish its output from uniform). Then, from the adversary's point of view, the source
essentially has min-entropy $k$.

Ideally, we would like to extract randomness from a weak source. However, it is easy to see that no deterministic
function can extract even one bit of randomness from all such sources, even for min-entropies as high as $n-1$ (see
e.g.~\cite{SV84}). One main approach to circumvent this problem is to use a short truly random {\em seed} for
extraction from the weak source ({\em seeded extractors}) (see, e.g., \cite{Shaltiel02}). The other main approach,
which is the focus of the current work, is to use several independent weak sources ({\em seedless extractors}) (e.g.
\cite{CG88,Vaz87,DEOR04,Bou05,Raz05} and many more).

With the advent of quantum computation, we must now deal with the possibility of quantum adversaries (or quantum
physical processes) interfering with the sources used for randomness extraction. For instance, one could imagine that a
quantum adversary now stores $n-k$ {\em qubits} of information about the string sampled from the source. This scenario
of a {\em bounded storage quantum adversary} arises in several applications, in particular in cryptography.

Some constructions of {\em seeded} extractors were shown to be secure in the presence of quantum adversaries:
K{\"o}nig, Maurer, and Renner~\cite{RK05, KMR05, Ren05} proved that the pairwise independent extractor of~\cite{ILL89}
is also good against quantum adversaries, and with the same parameters. K{\"o}nig and Terhal \cite{KT08} showed that
any one-bit output extractor is also good against quantum adversaries, with roughly the same parameters.
In light of this, it was tempting to conjecture that {\em any} extractor is also secure against quantum storage.
Somewhat surprisingly, Gavinsky et al.~\cite{GKKRW08} gave an example of a seeded extractor that is secure against
classical storage but becomes insecure even against very small quantum storage. This example has initiated a series of
recent ground-breaking work that examined which seeded extractors stay secure against bounded storage quantum
adversaries.
Ta-Shma~\cite{TaShma09} gave an extractor with a short (polylogarithmic) seed extracting a polynomial fraction of the
min-entropy. His result was improved by De and Vidick~\cite{DV10} extracting almost all of the min-entropy. Both
constructions are based on Trevisan's extractor~\cite{Tre01}.

However, the question of whether {\em seedless} multi-source extractors can remain secure against quantum adversaries
has remained wide open. The multi-source scenario corresponds to several independent adversaries, each tampering with
one of the sources, and then jointly trying to distinguish the extractor's output from uniform. In the classical
setting this leads to several independent weak sources. In the quantum world, measuring the adversaries' stored
information might break the independence of the sources, thus jeopardizing the performance of the
extractor.\footnote{Such an effect appears also in {\em strong seeded} extractors and has been discussed in more detail
in~\cite{KT08}.} Moreover, the multi-source setting offers a completely new aspect of the problem: the adversaries
could potentially share {\em entanglement} prior to tampering with the sources. Entanglement between several parties
has been known to yield several astonishing effects with no counterpart in the classical world, e.g., non-local
correlations~\cite{Bell:64a} and superdense coding~\cite{BW92}.

We note that the example of Gavinsky et al. can also be viewed as an example in the two-source model; we can imagine
that the seed comes from a second source (of full entropy in this case, just like any seeded extractor can be
artificially viewed as a two-source extractor). And obviously, in the same way, recent work on quantum secure seeded
extractors artificially gives secure two-source extractors, albeit for a limited range of parameters and without
allowing for entanglement. However, no one has as of yet explored how more realistic multi-source extractors fare
against quantum adversaries, and in particular how entanglement might change the picture. We ask: Are there any good
multi-source extractors secure against quantum bounded storage? And does this remain true when considering
entanglement?

\paragraph{Our results:} In this paper we answer all these questions in the positive. We focus on the inner-product
based two-source extractor of Dodis et al.~\cite{DEOR04} (DEOR-extractor).
Given two independent weak sources $X$ and $Y$ with the same length $n$ and min-entropies $k_1$ and $k_2$ satisfying
$k_1+k_2 \gtrapprox n$, this extractor gives $m$ close to uniform random bits, where $m \approx
\max(k_1,k_2)+k_1+k_2-n$. In recent years several two-source extractors with better parameters have been presented;
however, the DEOR-construction stands out through its elegance and simplicity and its parameters still fare very well
in comparison with recent work (e.g.,~\cite{Bou05,Raz05}).

A first conceptual step in this paper is to define the model of quantum adversaries and of security in the two-source
scenario (see Defs.~\ref{def:storage} and \ref{def:ext_storage}): Each adversary gets access to an independent weak
source $X$ (resp.~$Y$), and is allowed to store a {\em short} arbitrary quantum state.\footnote{In the setting of
seeded extractors with one source, this type of adversary was called {\em quantum encoding} in~\cite{TaShma09}.} In the
entangled setting, the two adversaries may share arbitrary prior entanglement, and hence their final joint stored state
is the possibly entangled state $\rho_{XY}$. In the non-entangled case their joint state is of the form
$\rho_{XY}=\rho_X \otimes \rho_Y$. In both cases, the security of the extractor is defined with respect to the joint
state they store.

\begin{definition}\label{def:ext_informal}[Two-source extractor against (entangled) quantum storage (informal):]
A function $E:\{0,1\}^n \times \{0,1\}^n \to \{0,1\}^m$ is a $(k_1,k_2,\eps)$ extractor against $(b_1,b_2)$ (entangled)
quantum storage if for any sources $X,Y$ with min-entropies $k_1,k_2$, and any joint stored quantum state  $\rho_{XY}$
prepared as above, with $X$-register of $b_1$ qubits and $Y$-register of $b_2$ qubits, the distribution $E(X,Y)$ is
$\eps$-close to uniform even when given access to  $\rho_{XY}$.
\end{definition}

Depending on the type of adversaries, we will say $E$ is secure against {\em entangled} or {\em non-entangled} storage.
Note again that entanglement between the adversaries is specific to the multi-source scenario and does not arise in the
case of seeded extractors.

Having set the framework, we show that the construction of Dodis et al.~\cite{DEOR04} is secure, first in the case of
non-entangled adversaries.

\begin{theorem}\label{thm:deor_unen}
The DEOR-construction is a $(k_1,k_2,\eps)$ extractor against $(b_1,b_2)$ non-entangled storage with $m =
(1-o(1))\max(k_1-\frac{b_1}{2},k_2-\frac{b_2}{2}) + \half(k_1-b_1+k_2-b_2-n) - 9\logeps - O(1)$ output bits, provided
$k_1+k_2-\max(b_1,b_2) > n + \Omega(\log^3(n/\eps))$.
\end{theorem}

As we show next the extractor remains secure even in the case of entangled adversaries. Notice the loss of essentially
a factor of $2$ in the allowed storage; this is related to the fact that superdense coding allows to store $n$ bits
using only $n/2$ entangled qubit pairs.

\begin{theorem}\label{thm:deor_en}
The DEOR-construction is a $(k_1,k_2,\eps)$ extractor against $(b_1,b_2)$ entangled storage with  $m =
(1-o(1))\max(k_1-b_2,k_2-b_1) + \half(k_1-2b_1+k_2-2b_2-n) - 9\logeps - O(1)$ output bits, provided $k_1 + k_2 -
2\max(b_1,b_2) > n + \Omega(\log^3(n/\eps))$.
\end{theorem}

Note that in both cases, when the storage is linear in the source entropy we can output $\Omega(n)$ bits with
exponentially small error.  To compare to the performance of the DEOR-extractor in the classical case, note that a
source with min-entropy $k$ and {\em classical} storage of size $b$ roughly corresponds to a source of min-entropy
$k-b$ (see, e.g.,~\cite{TaShma09} Lem. 3.1). Using this correspondence, the extractor of~\cite{DEOR04} gives $m =
\max(k_1,k_2) + k_1 - b_1 + k_2 - b_2 - n - 6\logeps - O(1)$ output bits against classical storage, whenever
$k_1+k_2-\max(b_1,b_2)>n+\Omega(\log n\cdot(\log^2n + \logeps))$. Hence the conditions under which we can extract
randomness are essentially the same for DEOR and for our Thm.~\ref{thm:deor_unen}. The amount of random bits we can
extract is somewhat less than in the classical case, even when disregarding storage.

In the non-entangled case, we are able to generalize our result to the stronger notion of guessing entropy adversaries
or so called {\em quantum knowledge} (see discussion below and Sec.~\ref{sec:guessing} for details). We show that the
DEOR-extractor remains secure even in this case, albeit with slightly weaker parameters.

\begin{theorem}\label{thm:deor_knowledge}
The DEOR-construction is a $(k_1,k_2,\eps)$ extractor against quantum knowledge with  $m = (1-o(1))\max(k_1,k_2) +
\frac{1}{6}(k_1+k_2-n) - 9\logeps - O(1)$ output bits, provided $k_1+k_2 > n + \Omega(\log^3(n/\eps))$.
\end{theorem}

{\em Strong extractors:} The extractor in Thms.~\ref{thm:deor_unen}, \ref{thm:deor_en} and \ref{thm:deor_knowledge} is
a so called {\em weak} extractor, meaning that when trying to break the extractor, no full access to any of the sources
is given (which is natural in the multi-source setting). We also obtain several results in the so called {\em strong}
case (see Cor.~\ref{cor:ext_ip_superstrong}, Lem.~\ref{lem:deor_xstrong_en}, Cor.~\ref{cor:ext_ip_knowledge} and
Lem.~\ref{lem:deor_xstrong_knowledge}). A {\em strong} extractor has the additional property that the output remains
secure even if the adversaries later gain full access to any one (but obviously not both) of the
sources.\footnote{In~\cite{DEOR04}, this is called a {\em strong blender}.} See Sec.~\ref{sec:prelim} for details and a
discussion of the subtleties in defining a strong extractor in the entangled case, and Secs.~\ref{sec:cc},
\ref{sec:many} and \ref{sec:guessing} for our results in the strong case.

{\em Tightness:} In the one-bit output case, we show that our results are {\em tight}, both in the entangled and
non-entangled setting (see Lem.~\ref{claim:ext_ip_tight}).

\paragraph{Proof ideas and tools:} To show both of our results, we first focus on the simplest case of one-bit outputs.
In this case the DEOR extractor \cite{DEOR04} simply computes the inner product  $E(x,y)=x \cdot y$ $\pmod 2$ of the
$n$-bit strings $x$ and $y$ coming from the two sources. Assume that the two adversaries are allowed quantum storage of
$b$ qubits each. Given their stored information they jointly wish to distinguish $E(x,y)$ from uniform, or, in other
words, to predict $x \cdot y$. We start by observing  that this setting corresponds to the well known simultaneous
message passing (SMP) model in communication complexity,\footnote{The connection between extractors and communication
complexity has been long known, see, e.g.,~\cite{Vaz87}.} where two parties, Alice and Bob, have access to an input
each (which is unknown to the other). They each send a message of length $b$ to a referee, who, upon reception of both
messages, is to compute a function $E(x,y)$ of the two inputs. When $E$ is hard to compute, it is a good extractor.
Moreover, the entangled adversaries case corresponds to the case of SMP with entanglement between Alice and Bob, a
model that has been studied in recent work (see e.g. \cite{GKRW09, GKW06}).

Before we proceed, let us remark, that there are cases, where entanglement is known to add tremendous power to the SMP
model. Namely, Gavinsky et al.~\cite{GKRW09} showed an exponential saving in communication in the entangled SMP model,
compared to the non-entangled case.\footnote{This result has been shown for a relation, not a function. It is tempting
to conjecture that this result can be turned into an exponential separation for an extractor with entangled vs.
non-entangled adversaries. It is, however, not immediate how to turn a worst case relation lower bound into an average
case function bound, as needed in the extractor setting, so we leave this problem open.} This points to the possibility
that some extractors can be secure against a large amount of storage in the non-entangled case, but be insecure against
drastically smaller amounts of entangled storage. Our results show that this is not the case for the DEOR extractor,
i.e., that this construction is secure against the potentially harmful effects of entanglement.

In the one-bit output DEOR case we can tap into known results on the quantum communication complexity of the inner
product problem (IP). Cleve et al.~\cite{CvDNT98} and Nayak and Salzman~\cite{NS06} have given tight lower bounds in
the one-way and two-way communication model, with and without entanglement (which also gives bounds in the SMP model).
For instance, in the non-entangled case, to compute IP exactly in the one-way model, $n$ qubits of communication are
needed, and in the SMP model, $n$ qubits of communication are needed from Alice and from Bob, just like in the
classical case. Note that whereas in the communication setting typically worst case problems are studied, extractors
correspond to {\em average case} (w.r.t. to weak randomness) problems. With some extra work we can adapt the
communication lower bounds to weak sources and to the average bias which is needed for the extractor result. In fact,
the results we obtain hold in the strong case (where later one of the sources is completely exposed), which corresponds
to one-way communication complexity.

Tightness of our results comes from matching upper bounds on the one-way and SMP model communication complexity of the
inner product. Adapting the work of \cite{CG88} we can obtain tight bounds for any bias $\eps$. Somewhat surprisingly,
it seems no one has looked at tight upper bounds for IP in the {\em entangled SMP model}, where \cite{CvDNT98} give an
$n/2$ lower bound for the message length for Alice and Bob. It turns out this bound is tight,\footnote{We thank Ronald
de Wolf~\cite{W:personal} for generously allowing us to adapt his upper bound to our setting.} which essentially leads
to the factor $2$ separation in our results for the entangled vs. non-entangled case (see Sec.~\ref{sec:cc}).

To show our results for the case of multi-bit extractors, we use the nice properties of the DEOR construction (and its
precursors \cite{Vaz87,DO03}). The extractor outputs bits of the form $Ax \cdot y$. Vazirani's XOR-Lemma allows to
reduce the multi-bit to the one-bit case by relating the distance from uniform of the multi-bit extractor to the sum of
biases of XOR's of subsets of its bits. Each such XOR, in turn, is just a (linearly transformed) inner product, for
which we already know how to bound the bias. Our main technical challenge is to adapt the XOR lemma to the case of {\em
quantum} side-information (see Sec.~\ref{sec:prelim}). This way we obtain first results for multi-bit extractors, which
even hold in the case of strong extractors. Following \cite{DEOR04}, we further improve the parameters in the {\em
weak} extractor setting by combining our strong two-source extractor with a good seeded extractor (in our case with the
construction of \cite{DPVR09}) to extract even more bits. See Sec.~\ref{sec:many} for details.

\paragraph{Guessing entropy:}
One can weaken the requirement of bounded storage, and instead only place a lower bound on the \emph{guessing entropy}
of the source given the adversary's storage, leading to the more general definition of extractors secure against
guessing entropy. Informally, a guessing entropy of at least $k$ means that the adversary's probability of correctly
guessing the source is at most $2^{-k}$ (or equivalently, that given the adversary's state, the source has essentially
min-entropy at least $k$). Working with guessing entropy has the advantage that we no longer have to worry about two
parameters (min-entropy and storage) instead only working with one parameter (guessing entropy), and that the resulting
extractors are stronger (assuming all other parameters are the same), see Sec.~\ref{sec:guessing}. In the classical
world, a guessing entropy of $k$ is more or less equivalent to a source with $k$ min-entropy; in the quantum world,
however, things become less trivial. In the case of seeded extractors, this more general model has been successfully
introduced and studied in~\cite{Ren05,KT08,FS08,DPVR09,TSSR10}, where several constructions secure against bounded
guessing entropy were shown.\footnote{Renner~\cite{Ren05} deals with the notion of {\em relative min-entropy}, which
was shown to be equivalent to guessing entropy~\cite{KRS09}.}

In the case of \emph{non-entangled} two-source extractors, we can show (based on~\cite{KT08}) that any classical {\em
one-bit} output two-source extractor remains secure against bounded guessing entropy adversaries, albeit with slightly
worse parameters.  Moreover, our XOR-Lemma allow us prove security of the DEOR-extractor against guessing entropy
adversaries even in the multi-bit case (Thm.~\ref{thm:deor_knowledge}, see Sec.~\ref{sec:guessing} for the
details).\footnote{We are grateful to Thomas Vidick for pointing out that our XOR-Lemma allows us to obtain results
also in this setting.}

In the {\em entangled} adversaries case, one natural way to define the model is to require the guessing entropy of each
source given the corresponding adversary's storage to be high. This definition, however, is too strong: it is easy to
see that no extractor can be secure against such adversaries. This follows from the observation that by sharing a
random string $r_1r_2$ (which is a special case of shared entanglement) and having the first adversary store $r_1
\oplus x,r_2$ and the other store $r_1,r_2 \oplus y$, we keep the guessing entropy of $X$ (resp. $Y$) relative to the
adversary's storage unchanged yet we can recover $x$ and $y$ completely from the combined storage.

Hence we are naturally lead to consider the weaker requirement that the guessing entropy of each source given the
combined storage of \emph{both} adversaries is high. We now observe that already the DEOR one-bit extractor (where the
output is simply the inner product) is not secure under this definition, indicating that this definition is still too
strong. To see this, consider uniform $n$-bit sources $X,Y$, and say Alice stores $x \oplus r$, and Bob stores $y
\oplus r$, where $r$ is a shared random string. Obviously, their joint state does not help in guessing $X$ (or $Y$),
hence the guessing entropy of the sources is still $n$; but their joint state does give $x \oplus y$. If, in addition,
Alice also stores the Hamming weight $|x|\bmod 4$ and Bob $|y|\bmod 4$, the guessing entropy is barely affected, and
indeed one can easily show it is $n-O(1)$. However, their information now suffices to compute $x\cdot y$ exactly, since
$x \cdot y=\half((|x|+|y|-|x \oplus y|)\bmod 4)$. Hence inner product is insecure in this model even for very high
guessing entropies, even though it is secure against a fair amount of bounded storage.

In light of this, it is not clear if and how entangled guessing entropy sources can be incorporated into the model, and
hence we only consider bounded storage adversaries in the entangled case.

\paragraph{Related work:} We are the first to consider two-source extractors in the quantum world, especially against entanglement.
As mentioned, previous work on seeded extractors against quantum adversaries~\cite{RK05, KMR05, Ren05, KT08, TaShma09,
DV10, DPVR09, BT10} gives rise to trivial two-source extractors where one of the sources is not touched by the
adversaries. However, the only previous work that allows to derive results in the genuine two-source scenario is the
work by K\"onig and Terhal \cite{KT08}. Using what is implicit in their work, and with some extra effort, it is
possible to obtain results in the one-bit output non-entangled two-source scenario (which hold against guessing entropy
adversaries, but with worse performance than our results for the inner product extractor), and we give this result in
detail in Sec.~\ref{sec:guessing}. Moreover, \cite{KT08} show that any classical multi-bit extractor is secure against
bounded storage adversaries, albeit with an exponential decay in the error parameter. This easily extends to the
non-entangled two-source scenario, to give results in the spirit of Thm.~\ref{thm:deor_unen}. We have worked out the
details and comparison to Thm.~\ref{thm:deor_unen} in App.~\ref{app:storage}. Note, however, that to our knowledge no
previous work gives results in the entangled scenario.

\paragraph{\bf Discussion and Open Problems:}
We have, for the first time, studied two-source extractors in the quantum world. Previously, only seeded extractors
have been studied in the quantum setting. In the two-source scenario a new phenomenon appears: entanglement between the
(otherwise independent) sources. We have formalized what we believe the strongest possible notion of quantum
adversaries in this setting and shown that one of the best performing extractors, the DEOR-construction, remains
secure. We also show that our results are tight in the one-bit output case.

Our results for the multi-bit output DEOR-construction allow to extract slightly less bits compared to what is possible
classically. An interesting open quesiotn is whether it is possible to obtain matching parameters in the
(non-entangled) quantum case. One might have to refine the analysis and not rely solely on communication complexity
lower bounds. Alternatively, our quantum XOR-Lemma currently incurs a penalty exponential in either the length of the
output or the length of the storage. Any improvement here also immediately improves all three main theorems. In
particular, by removing the penalty entirely, Thm.~\ref{thm:deor_unen} can be made essentially optimal (with respect to
the classical case).

We have shown that inner product based constructions are necessarily insecure in two reasonable models of entangled
guessing entropy adversaries (and hence that bounded storage adversaries are the more appropriate model in the
entangled case). It should be noted that it is possible that other extractor constructions (not based on inner product)
could remain secure in this setting, and this subject warrants further exploration.

As pointed out, it is conceivable that entanglement could break the security of two-source extractors. Evidence for
this is provided by the communication complexity separation in the entangled vs. non-entangled SMP-model, given in
\cite{GKRW09}. A fascinating open problem is to turn this relational separation into an extractor that is secure
against non-entangled quantum adversaries but completely broken when entanglement is present.

Our work leaves several other open questions. It would be interesting to see if other multi-source extractors remain
secure against entangled adversaries, in particular the recent breakthrough construction by Bourgain \cite{Bou05} which
works for two sources with min-entropy $(1/2-\alpha)n$ each for some small constant $\alpha$, or the construction of
Raz \cite{Raz05}, where one source is allowed to have logarithmic min-entropy while the other has min-entropy slightly
larger than $n/2$. Both extractors are based on the inner product and output $\Omega(n)$ almost uniform bits.

And lastly, it would be interesting to see other application of secure multi-source extractors in the quantum world.
One possible scenario is multi-party computation. Classically, Kalai et al.~\cite{KLR09} show that sufficiently strong
two-source extractors allow to perform multi-party communication with weak sources when at least two parties are
honest. Perhaps similar results hold in the quantum setting.

\paragraph{Structure of the paper:}
In Sec.~\ref{sec:prelim} we introduce our basic notation and definitions, and describe the DEOR construction. Here we
also present one of our tools, the "quantum" XOR-Lemma. Sec.~\ref{sec:cc} is dedicated to the one-bit output case and
the connection to communication complexity and gives our tightness results. In Sec.~\ref{sec:many} we deal with the
multi-bit output case and prove our main result, Thms.~\ref{thm:deor_unen} and \ref{thm:deor_en}. In
Sec.~\ref{sec:guessing} we present our results against non-entangled guessing entropy adversaries (partly based
on~\cite{KT08}) and prove Thm.~\ref{thm:deor_knowledge}. App.~\ref{app:storage} works out the results that can be
derived from \cite{KT08} in the case of multi-bit extractors against non-entangled bounded storage.

\section{Preliminaries and Tools}\label{sec:prelim}
In this section we provide the necessary notation, formalize Def.\ref{def:ext_informal}, describe the DEOR-extractor
and present and prove our quantum XOR-Lemma. For background on quantum information see e.g. \cite{NC00}.

\paragraph{Notation:}
Given a classical random variable $Z$ and a set of density matrices $\set{\rho_z}_{z\in Z}$ we denote by $Z\rho_Z$ the
cq-state $\sum_{z\in Z} \Pr[Z=z]\ketbra{z}{z}\tensor \rho_z$. When the distribution is clear from the context we write
$p(z)$ instead of $\Pr[Z=z]$. For any random variable $Z'$ on the domain of $Z$, we define $\rho_{Z'}: = \sum_{z\in Z'}
\Pr[Z'=z]\rho_z$. For any random variable $Y$, let $Y\rho_Z:=\sum_{y\in Y} \Pr[Y=y]\ketbra{y}{y}\tensor \rho_{Z|Y=y}$.
We denote by $U_m$ the uniform distribution on $m$ bits. For matrix norms, we define $\norm{tr}{A} = \half \dnorm{1}{A}
= \half \tr{\sqrt{A^\dagger A}}$ and $\dnorm{2}{A} = \sqrt{\tr{A^\dagger A}}$.

\paragraph{Extractors against quantum storage:} We first formalize the different types of quantum storage.
\begin{definition}\label{def:storage}
For two random variables $X,Y$ we say $\rho_{XY}$ is a $(b_1,b_2)$ entangled storage if it is generated by two non
communicating parties, Alice and Bob, in the following way. Alice and Bob share an arbitrary entangled state. Alice
receives $x \in X$, Bob receives $y \in Y$. They each apply any quantum operation on their qubits. Alice then stores
$b_1$ of her qubits (and discards the rest), and Bob stores $b_2$ of his qubits, giving the state $\rho_{xy}$.

We denote by $\rho^A_{XY}$ the state obtained when Alice stores her entire state, whereas Bob stores only $b_2$ qubits
of his, and similarly for $\rho^B_{XY}$.

We say $\rho_{XY}$ is $(b_1,b_2)$ non-entangled storage if $\rho_{xy} = \rho_x \tensor \rho_y$ for all $x \in X,y \in
Y$.
\end{definition}

The security of the extractor is defined relative to the storage.

\begin{definition}\label{def:ext_storage}
A $(k_1, k_2, \eps)$ 2-source extractor against $(b_1,b_2)$ (entangled) quantum storage is a function $E:\zo^n \times
\zo^n \rightarrow \zo^m$ such that for any independent $n$-bit weak sources $X,Y$ with respective min-entropies
$k_1,k_2$, and any $(b_1, b_2)$ (entangled) storage $\rho_{XY}$, $\norm{tr}{E(X,Y)\rho_{XY} - U_m\rho_{XY}} \le \eps$.

The extractor is called {\em X-strong} if $\norm{tr}{E(X,Y)\rho_{XY}X - U_m\rho_{XY}X} \le \eps$,
and {\em X-superstrong} when $\rho_{XY}$ is replaced by $\rho^A_{XY}$.
It is called {\em (super)strong} if it is both X- and Y- {\em (super)strong}.
\end{definition}

A note on the definition: A strong extractor is secure even if at the distinguishing stage one of the sources is
completely exposed.
A superstrong extractor is secure even if, in addition, the matching party's entire state is also given. Without entanglement, the two are equal, as the state can be completely reconstructed from the source. In the
communication complexity setting the model of strong extractors corresponds to the SMP model where the referee also
gets access to one of the inputs, whereas the model of superstrong extractors corresponds to the one-way model,
where one party also has access to its share of the entangled state.

To prove $E$ is an extractor, it suffices to show that it is either X-strong or Y-strong. All our proofs follow this
route.

\paragraph {\bf Flat sources:} It is well known that any source with min-entropy $k$ is a convex combination of flat sources (i.e., sources that
are uniformly distributed over their support) with min-entropy $k$. In what follows we will therefore only consider
such sources in our analysis of extractors, as one can easily verify that  \[ \norm{tr}{E(X,Y)\rho_{XY} - U_m\rho_{XY}}
\le \max_{i,j} \norm{tr}{E(X_i,Y_j)\rho_{X_iY_j} - U_m\rho_{X_iY_j}},
\]
where $X = \sum \alpha_i X_i$ and $Y = \sum \beta_j Y_j$ are convex combinations of flat sources.

\paragraph{The DEOR construction:}
The following (strong) extractor construction is due to Dodis et al. \cite{DEOR04}. Every output bit is a linearly
transformed inner product, namely $A_ix\cdot y$ for some full rank matrix $A_i$, where $x$ and $y$ are the $n$-bit
input vectors. Here $x \cdot y:= \sum_{j=1}^n x_j y_j \pmod 2$. The matrices $A_i$ have the additional property that
every subset sum is also of full rank. This ensures that any XOR of some bits of the output is itself a linearly
transformed inner product.

\begin{lemma}[\cite{DEOR04}]
For all $n>0$, there exist an efficiently computable set of $n \times n$ matrices $A_1,A_2,\ldots,A_n$ over GF(2) such
that for any non-empty set $S \subseteq [n]$, $A_S:=\sum_{i\in S} A_i$ has full rank.
\end{lemma}

\begin{definition}[strong blender of \cite{DEOR04}]\label{def:DEOR}
Let $n \ge m > 0$, and let $\set{A_i}_{i=1}^m$ be a set as above. The DEOR-extractor $E_D:\zo^n \times \zo^n
\rightarrow \zo^m$ is given by $E_D(x,y) = A_1x\cdot y, A_2x\cdot y,\ldots,A_mx\cdot y$.
\end{definition}

\paragraph{The XOR-Lemma:} Vazirani's XOR-Lemma \cite{Vaz87} relates the non-uniformity of a distribution to the non-uniformity of the characters
of the distribution, i.e., the XOR of certain bit positions. For the DEOR-extractor it allows to reduce the multi-bit
output case to the binary output case.
\begin{lemma}[Classical XOR-Lemma \cite{Vaz87,Gol95}]
For every $m$-bit random variable $Z$ \[
\norm{1}{Z - U_m}^2 \le \sum_{0 \neq S \in \zo^m} \norm{1}{(S\cdot Z) - U_1}^2. \]
\end{lemma}

This lemma is not immediately applicable in our scenario, as we need to take into account {\em quantum} side
information. For this, we need a slightly more general XOR-Lemma.

\begin{lemma}[Classical-Quantum XOR-Lemma]\label{lem:xor}\footnote{We thank Thomas Vidick for pointing out that we can
also have a bound in terms of $m$ and not only $d$.} Let $Z\rho_Z$ be an arbitrary cq-state, where $Z$ is an $m$-bit
classical random variable and $\rho_Z$ is of dimension $2^d$. Then \[ \norm{tr}{Z\rho_Z - U_m \rho_Z}^2 \le
2^{\min(d,m)} \cdot \sum_{0 \neq S \in \zo^m} \norm{tr}{(S\cdot Z)\rho_Z - U_1\rho_Z}^2.
\]
\end{lemma}

\begin{proof}
Following the proof of the classical XOR-Lemma in \cite{Gol95}, we first relate $\dnorm{1}{Z\rho_Z - U_m \rho_Z}$ to
$\dnorm{2}{Z\rho_Z - U_m \rho_Z}$, and then view $Z\rho_Z - U_m \rho_Z$ in the Hadamard (or Fourier) basis, giving us
the desired result. We need the following simple claim.

\begin{claim}\label{claim:xor_pre}
For any Boolean function $f$, $\dnorm{1}{f(Z)\rho_Z - U_1\rho_Z} =
\dnorm{1}{\sum_z (-1)^{f(z)}p(z)\rho_z}$.
\end{claim}
\begin{proof}
Denote $\rho_b = \sum_{z:f(z)=b} p(z)\rho_z$ for $b=0,1$. Then $\rho_Z = \rho_0 + \rho_1$ and
\begin{align}
\dnorm{1}{f(Z)\rho_Z - U_1\rho_Z} & =
    \dnorm{1}{\ketbra{0}{0}\tensor\rho_0 + \ketbra{1}{1}\tensor\rho_1 - \half(\ketbra{0}{0} + \ketbra{1}{1})\tensor(\rho_0 + \rho_1)} \notag \\
 & = \half\dnorm{1}{\ketbra{0}{0}\tensor(\rho_0-\rho_1) + \ketbra{1}{1}\tensor(\rho_1-\rho_0)}  \notag  \\
 & = \dnorm{1}{\rho_0-\rho_1}=\dnorm{1}{\sum_z (-1)^{f(z)}p(z)\rho_z}. \label{eq:trace_one_bit}
\end{align}
\end{proof}
Let $\chi_S(z) = (-1)^{S \cdot z}$ for $S \in \zo^m$. Denote $D=2^d$, $M=2^m$, and $\sig_z = p(z)\rho_z - \frac{1}{M}\rho_Z$. Then \begin{align}
& \dnorm{1}{Z\rho_Z - U_m\rho_Z}^2 = \dnorm{1}{\sum_z \ketbra{z}{z}\tensor \sig_z}^2
 = \dnorm{1}{(H^{\otimes m} \tensor I_D)\left(\sum_z \ketbra{z}{z}\tensor \sig_z\right)(H^{\otimes m} \tensor I_D)}^2 \notag \\
& = \frac{1}{M^2}\cdot \dnorm{1}{\sum_{z,y,S} \ketbra{y}{S}\otimes \chi_S(z)\chi_y(z)\sig_z}^2
\le \frac{D}{M}\cdot \dnorm{2}{\sum_{z,y,S} \ketbra{y}{S}\otimes \chi_S(z)\chi_y(z)\sig_z}^2 \label{eq:XOR1},
\end{align}
where $H$ is the Hadamard transform.

\paragraph{Factor $D$:}
Using the fact that the $\dnorm{2}{\cdot}^2$ of a matrix is the sum of $\dnorm{2}{\cdot}^2$ of its ($D \times D$)
sub-blocks, together with $\chi_S(z)\chi_y(z)=\chi_{y+S}(z)$ and $\dnorm{2}{\cdot} \leq \dnorm{1}{\cdot}$,
\eqref{eq:XOR1} gives
\begin{align} \dnorm{1}{Z\rho_Z - U_m\rho_Z}^2 \le \frac{D}{M}\sum_y \sum_S
\dnorm{2}{\sum_z \chi_{y+S}(z)\sig_z}^2
 = D \sum_S \dnorm{2}{\sum_z \chi_S(z)\sig_z}^2 \le D
\sum_S \dnorm{1}{\sum_z \chi_S(z)\sig_z}^2. \label{eq:XOR2}
\end{align}
Using Claim \ref{claim:xor_pre} with $f(Z)={S \cdot Z}$, we get
\begin{align}  \sum_{S \neq 0} \dnorm{1}{(S\cdot Z)\rho_Z - U_1\rho_Z}^2 &= \sum_{S \neq 0} \dnorm{1}{\sum_z
\chi_S(z)p(z)\rho_z}^2
 = \sum_{S \neq 0} \dnorm{1}{\sum_z \chi_S(z)\sig_z}^2= \sum_S \dnorm{1}{\sum_z \chi_S(z)\sig_z}^2,
 \label{eq:XOR3}
\end{align}
where the second equality holds since $\chi_S$ is balanced, and the third since $\sum_z \sig_z = 0$. Combining Eqs.
\eqref{eq:XOR2} and \eqref{eq:XOR3} gives the desired result.

\paragraph{Factor $M$:} Restarting from the next-to-last step of \eqref{eq:XOR1}, using again $\chi_S(z)\chi_y(z)=\chi_{y+S}(z)$ and the
triangle inequality, we obtain
\begin{align*}
\dnorm{1}{Z\rho_Z - U_m\rho_Z}^2 &\leq \frac{1}{M^2} \cdot \left(\sum_S \dnorm{1}{\sum_{y}\ketbra{y}{S+y}\otimes
\left(\sum_z\chi_{S}(z)\sig_z\right)}\right)^2 \\
&\le \frac{1}{M} \cdot \sum_S \dnorm{1}{\sum_{y}\ketbra{y}{S+y}\otimes \left(\sum_z\chi_{S}(z)\sig_z\right)}^2 =
M\cdot\sum_S \dnorm{1}{\sum_z\chi_{S}(z)\sig_z}^2,
\end{align*}
where the last step follows from the observation that the matrices inside the norms are of the form $P \tensor B$ where
$P$ is a permutation matrix. In this case $\dnorm{1}{P \tensor B} = \dim(P)\cdot\dnorm{1}{B}=M \cdot \dnorm{1}{B}$. As
before, combining this with Eq. \eqref{eq:XOR3} gives the desired bound.
\end{proof}

\section{Communication Complexity and One-Bit Extractors}\label{sec:cc}

\subsection{Average case lower bound for inner product}

Cleve et al. \cite{CvDNT98} give a lower bound for the worst case one-way quantum communication complexity of inner
product with arbitrary prior entanglement. It is achieved by first reducing the problem of computing the inner product
to that of transmitting one input over a quantum channel, and then using an extended Holevo bound. Nayak and Salzman
\cite{NS06} obtained an optimal lower bound by replacing Holevo with a more ``mission-specific" bound:

\begin{theorem}[\cite{NS06}, Thm 1.3 and discussion thereafter]\label{thm:ns}
Let $X$ be an $n$-bit random variable with min-entropy $k$, and suppose Alice wishes to convey $X$ to Bob over a
one-way quantum communication channel using $b$ qubits. Let $Y$ be the random variable denoting
Bob's guess for $X$. Then \begin{enumerate} \itemsep=-1pt \item $\Pr[Y = X] \le 2^{-(k-b)}$, if the parties don't share
prior entanglement, and  \label{thm_item:ns_unen} \item $\Pr[Y = X] \le 2^{-(k-2b)}$. \label{thm_item:ns_en}
\end{enumerate}
\end{theorem}
Revisiting Cleve et al.'s reduction, we now show how to adapt it to flat sources, to the average case error and to the
linearly transformed inner product. The main challenge is to carefully treat the error terms so as to not cancel out
the (small) amplitude of the correct state.

\begin{lemma}\label{lem:enhanced_cleve}
Let $X,Y$ be flat sources over $n$ bits with min-entropies $k_1,k_2$, and $A,B$ full rank $n$ by $n$ matrices over
$GF(2)$. Let $P$ be a $b$ qubit one-way protocol for $(AX)\cdot (BY)$ with success probability $\half + \eps$. Then
\begin{enumerate}[\indent(a)] \itemsep=-1pt
\item $\eps \le 2^{-(k_1+k_2-2b-n+2)/2}$, if the parties share prior entanglement and \label{lem_item:cl_en}
\item $\eps \le 2^{-(k_1+k_2-b-n+2)/2}$ otherwise. \label{lem_item:cl_unen}
\end{enumerate}
\end{lemma}
\begin{proof}
Let us first consider the case $A=B=I$. Assume w.l.o.g. Bob delays his operations until receiving the message from
Alice and that in his first step he copies his input, leaving the original untouched throughout. Further assume Bob
outputs the result in one of his qubits.

For a fixed $x$, denote the success probability of $P$ by $\half + \eps_x$ ($\eps_x$ might be negative).
Denote Bob's state after receiving the message as $\ket{y}\ket{0}\ket{\sig_x}$, where $\sig_x$ is taken to contain
Alice's message and Bob's prior entangled qubits as required by the protocol (if present). The rest of the protocol is
now performed locally by Bob. We denote this computation $P_B$. After applying $P_B$, Bob's state is of the form
\[ \alpha_{x,y}\ket{y}\ket{x\cdot y}\ket{J_{x,y}} +
\beta_{x,y}\ket{y}\ket{\overline{x\cdot y}}\ket{K_{x,y}}, \] and by assumption, $\E_y \beta_{x,y}^2 = \half - \eps_x$.
Following the analysis in \cite{CvDNT98}, using {\em clean} computation, where the output is produced in a new
qubit (the leftmost), gives the state \[ \ket{z + x\cdot
y}\ket{y}\ket{0}\ket{\sig_x} + \sqrt{2}\beta_{x,y}\ket{M_{x,y,z}}, \] where $\ket{M_{x,y,z}} =
\left(\frac{1}{\sqrt{2}}\ket{z + \overline{x\cdot y}} - \frac{1}{\sqrt{2}}\ket{z + x\cdot y}\right)
P_B^{\dagger}\ket{y}\ket{\overline{x\cdot y}}\ket{K_{x,y}}$. Observe the following properties of $M$: 1.
$\ket{M_{x,y,0}} = -\ket{M_{x,y,1}}$ 2. As $y \in Y$ varies, the states $\ket{M_{x,y,z}}$ are orthonormal. 3. Since
$P_B^{\dagger}$ does not affect the first $n$ (so called input) qubits, $\ket{M_{x,y,z}}$ is orthogonal to states of
the form $\ket{a}\ket{y'}\tensor\ket{\cdot}$ for all $a \in \zo, y \in Y, y' \notin Y$.

We now use the following steps to transfer $X$ from Alice to Bob:
\begin{enumerate} \itemsep=-1pt
\item Bob prepares the state $\sqrt{2^{-k_2-1}}\cdot\sum_{y \in Y, a \in \zo}(-1)^{a}\ket{a}\ket{y}$. \item Alice and
Bob execute the clean version of $P$. \item Bob performs the Hadamard transform on each of his first $n+1$ qubits and
measures in the computational basis.
\end{enumerate}
After the second step, Bob's state is $\ket{\psi} = \ket{v} + \vec{e}$ where \begin{align*} \ket{v} =
\sqrt{2^{-k_2-1}}\sum_{y \in Y, a \in \zo}(-1)^{a+x\cdot y}\ket{a}\ket{y}\ket{0}\ket{\sig_x} & & \vec{e} =
\sqrt{2^{-k_2-1}}\sum_{y \in Y, a \in \zo}(-1)^a\sqrt{2}\beta_{x,y}\ket{M_{x,y,a}}.
 \end{align*}
By the properties of $\ket{M_{x,y,z}}$, $\left\|{\vec{e}}\right\| = 2\sqrt{\E_y\beta_{x,y}^2} = 2\sqrt{\half -
\eps_x}$. Since $\ket{v} + \vec{e}$ and $\ket{v}$ are normalized states, we can easily derive $\bra{v}(\ket{v} +
\vec{e}) = 2\eps_x$. Define
\[ \ket{\psi_0} = H^{\tensor n+1}\ket{1x}\tensor\ket{0}\ket{\sig_x} = \sqrt{2^{k_2-n}}\ket{v} + \sqrt{2^{-n-1}}\sum_{y
\notin Y, a\in\zo}(-1)^{a+x\cdot y}\ket{a}\ket{y}\ket{0}\ket{\sig_x},
\] and note that the second term is orthogonal to both $\ket{v}$ and $\vec{e}$. It follows that $\braket{\psi}{\psi_0}
= \sqrt{2^{k_2-n+2}}\eps_x$. Applying the Hadamard transform  in  Step 3. does not affect the inner product, and so Bob
will measure $\ket{1x}$ with probability $2^{k_2-n+2}\cdot\eps_x^2$. Applying Thm.~\ref{thm:ns}.\ref{thm_item:ns_unen}
and \ref{thm:ns}.\ref{thm_item:ns_en} along with Jensen's inequality now completes the proof.

For the general case where $A \neq I$ or $B \neq I$, we modify Step 3.~of the transmission protocol. Instead of the
Hadamard transform, Bob applies the inverse of the unitary transformation $\ket{z}\ket{x} \mapsto
\sqrt{2^{-n-1}}\cdot\sum_{y,a}(-1)^{za+(Ax)\cdot(By)}\ket{a}\ket{y}$. It is easy to check that this gives the desired
result.
\end{proof} % of cleve_strong

\subsection{One bit extractor}\label{subsec:one_bit}

When the extractor's output is binary, distinguishing it from uniform is equivalent to computing the output on average.
This was shown by Yao \cite{Yao82} when the storage is classical and is trivially extended to the quantum setting. With
this observation, reformulating Lem.~\ref{lem:enhanced_cleve} in the language of trace distance yields a one bit
extractor.

\begin{corollary}\label{cor:ext_ip}
The function $E_{IP}(x,y) = x \cdot y$ is a $(k_1,k_2,\eps)$ extractor against $(b_1,b_2)$ (entangled) quantum storage
provided \begin{enumerate}[\indent(a)] \itemsep=-1pt
\item (entangled) $k_1+k_2-2\min(b_1,b_2) \ge n-2+2\logeps$,
\item (non-entangled) $k_1+k_2-\min(b_1,b_2) \ge n-2+2\logeps$.
\end{enumerate}
\end{corollary}
\begin{proof}
With Yao's equivalence, Lem.~\ref{lem:enhanced_cleve}.\eqref{lem_item:cl_en} immediately gives
\begin{align}
\norm{tr}{(AX\cdot Y)\rho_{XY}X - U\rho_{XY}X} \le 2^{-{(k_1+k_2-2b_2-n+2)/2}} \label{pty:ext_cleve_x}\\
\norm{tr}{(AX\cdot Y)\rho_{XY}Y - U\rho_{XY}Y} \le 2^{-{(k_1+k_2-2b_1-n+2)/2}} \label{pty:ext_cleve_y}
\end{align}
for any full rank matrix $A$, and specifically for $A=I$. By the assumption on $\eps$, $E_{IP}$ is either Y-strong or
X-strong. Repeating this argument with Lem.~\ref{lem:enhanced_cleve}.\eqref{lem_item:cl_unen} gives the non-entangled
case.
\end{proof}

Recall (see Def.~\ref{def:ext_storage} and discussion thereafter) that one-way communication corresponds to the model
of {\em superstrong} extractors. It is not surprising then that Lem.~\ref{lem:enhanced_cleve} actually implies a
superstrong extractor. By choosing $\eps$ in the above proof of Cor.~\ref{cor:ext_ip} such that both inequalities
\eqref{pty:ext_cleve_x} and \eqref{pty:ext_cleve_y} are satisfied, where we replace $\rho_{xY}$ by $\rho^A_{xY}$ to
include Alice's complete state as well as Bob's entangled qubits and similarly for $\rho^B_{Xy}$, we obtain:

\begin{corollary}\label{cor:ext_ip_superstrong}
The function $E_{IP}(x,y) = x \cdot y$ is a $(k_1,k_2,\eps)$ {\em superstrong} extractor against $(b_1,b_2)$
(entangled) quantum storage provided \begin{enumerate}[\indent(a)] \itemsep=-1pt
\item (entangled) $k_1+k_2-2\max(b_1,b_2) \ge n-2+2\logeps$,
\item (non-entangled) $k_1+k_2-\max(b_1,b_2) \ge n-2+2\logeps$.
\end{enumerate}
\end{corollary}

We now show that the parameters of all our extractors are {\em tight} up to an additive constant.
For simplicity, assume first that the error $\eps$ is close to $1/2$, the sources are uniform and
$b_1=b_2:=b$. Cor.~\ref{cor:ext_ip} then states that $E_{IP}$ is an extractor as long as $b < n$ in the non-entangled
case and $b < n/2$ in the entangled case. Indeed, in the non-entangled case it is trivial to compute the inner product
in the SMP model (i.e., break the extractor) when $b \ge n$. With entanglement, $b \ge n/2$ suffices as demonstrated by
the following protocol, adapted from a protocol by de Wolf \cite{W:personal}.

\begin{claim}\label{claim:ip_smp_protocol}
The inner product function for $n$ bit strings is exactly computable in the SMP model with entanglement with $n/2+2$
qubits of communication from each party.
\end{claim}
\begin{proof}
Let $x,y\in \zo^n$ be Alice and Bob's inputs. Since $x \cdot y = \half((|x|+|y|-|x \oplus y|)\bmod 4)$, it suffices to
show that the referee can compute $x \oplus y$ with $n/2$ qubits of communication from each party, or simply $x_1x_2
\oplus y_1y_2$ with one qubit of communication each.

Denote the Pauli matrices $\sig_{00} = I$, $\sig_{01} = Z$, $\sig_{10} = X$, $\sig_{11} = ZX$. Given a shared EPR pair,
Alice applies $\sig_{x_1x_2}$ to her qubit and sends it to the referee, and Bob does the same with
$\sig_{y_1y_2}$. Note that applying $\sig_{b_1b_2}$ to the first qubit has the same effect as applying it to the second
qubit. Further, $X$ is applied iff $b_1$ is 1 and $Z$ is applied iff $b_2$ is 1. Since two applications of $X$ ($Z$)
cancel each other out, we have that $X$ is applied to the first qubit iff $x_1+y_1=1$ and $Z$ is applied to the first
qubit iff $x_2+y_2=1$. The net effect on the EPR state is $\sig_{x_1x_2\oplus y_1y_2}\tensor I$. For each value of
$x_1x_2 \oplus y_1y_2$ this gives one of the orthogonal (completely distinguishable) Bell states.
\end{proof}

Showing that our results are tight for arbitrary $\eps$ is trickier. We show

\begin{lemma}\label{claim:ext_ip_tight}
If $E_{IP} = x\cdot y$ is a $(k_1,k_2,\eps)$ extractor against $(b_1, b_2)$ (entangled) storage then
\begin{enumerate}[\indent(a)] \itemsep=-1pt
\item (entangled) $k_1 + k_2 - 2\min(b_1,b_2) > n - 9 + 2\logeps$,
\item (non-entangled) $k_1 + k_2 - \min(b_1,b_2) > n - 5 + 2\logeps$.
\end{enumerate}
If $E_{IP}$ is {\em superstrong}, then
\begin{enumerate}[\indent(a)] \itemsep=-1pt
\item (entangled) $k_1 + k_2 - 2\max(b_1,b_2) > n - 9 + 2\logeps$,
\item (non-entangled) $k_1 + k_2 - \max(b_1,b_2) > n - 5 + 2\logeps$.
\end{enumerate}
\end{lemma}
\begin{proof} We give a slightly modified version of Proposition 10 in \cite{CG88}, taking into account quantum side information. We
need the following theorem.

\begin{theorem}[{\cite[Theorem 3]{CG88}}]\label{thm:hadamard_bias}
There exist independent random variables $X,Y$ on $l$ bits with min-entropy $l-3$ each\footnote{\cite{CG88} prove the
claim with slightly different parameters for arbitrary Boolean functions. Our modification is trivial.} such that
$\Pr[X\cdot Y = 0] > \half + 2^{-(l-1)/2}.$
\end{theorem}

We start in the weak extractor setting with entanglement.
We construct sources $X,Y$ with min-entropy $k_1,k_2$ and $(b_1,b_2)$ entangled quantum storage $\rho_{XY}$ for which
the error will be "large". Let $b = 2(\min(b_1,b_2)-2)$, and let $\Delta = k_1+k_2-n$. If $\Delta \le b$, we pick $X$
to be uniform on the first $k_1$ bits and $0$ elsewhere, $Y$ uniform on the last $k_2$ bits and $0$ elsewhere. The
inner product of $X,Y$ is then the inner product of at most $b$ bits, and can be computed exactly using the SMP
protocol in Claim \ref{claim:ip_smp_protocol} with $\min(b_1,b_2)$ qubits from each.

In the case $\Delta > b$, we define $X = X_1 X_2 X_3 X_4$ as follows: $X_1$ is uniform on $b$ bits, $X_2$ is uniform on
$k_1-\Delta-3$ bits, $X_3$ is the first $(\Delta+6-b, \Delta+3-b)$ source promised by Thm.~\ref{thm:hadamard_bias} (for
$l=\Delta+6-b$), and $X_4$ is constant $0^{n-k_1-3}$. Analogously, $Y = Y_1 Y_2 Y_3 Y_4$ is defined as: $Y_1$ is
uniform on $b$ bits, $Y_2$ is constant $0^{n-k_2-3}$, $Y_3$ is the second $(\Delta+6-b, \Delta+3-b)$ source promised by
Thm.~\ref{thm:hadamard_bias}, and $Y_4$ is uniform on $k_2-\Delta-3$ bits. It is easily verified that $H_\infty(X) \ge
k_1$ and $H_\infty(Y) \ge k_2$. Finally, we set $\rho_{XY}$ to be the entangled $(\min(b_1,b_2), \min(b_1,b_2))$
storage of the SMP protocol in Claim \ref{claim:ip_smp_protocol} allowing us to compute $x_1 \cdot y_1$ exactly, and
$M$ the measurement strategy of the referee. Applying Thm.~\ref{thm:hadamard_bias}, \[ \Pr[M(\rho_{XY}) = X\cdot Y] =
\Pr[X_1\cdot Y_1 = X\cdot Y] = \Pr[X_3\cdot Y_3 = 0] > \half+2^{-(\Delta+5-b)/2} \] and $\norm{tr}{(X\cdot Y)\rho_{XY}
- U\rho_{XY}} > 2^{-(k_1+k_2-b-n+5)/2}$.

In the non-entangled case, we simply set $b = \min(b_1,b_2)$ and replace the SMP protocol with a trivial protocol for
IP on $b$ bits.\footnote{In fact, this shows that our non-entangled extractor is tight even for {\em classical}
storage.}

In the superstrong case with entanglement, assume w.l.o.g. that $b_1 > b_2$ and choose $b = b_1/2$. We then let
$\rho_{xy}$ be  the entangled state that appears in the superdense coding protocol for $X_1$. Thus, exposing Bob's
state allows us to compute $X_1\cdot Y_1$ exactly. Without entanglement, we set $b = b_1$ and have Alice send $X_1$ to
Bob.
\end{proof}

\section{Many Bit Extractors}\label{sec:many}

Here we prove our main Theorems \ref{thm:deor_unen} and \ref{thm:deor_en}. First, using our quantum XOR-Lemma
\ref{lem:xor}, we obtain results in the {\em strong} case.

\begin{lemma}\label{lem:deor_xstrong_en}
$E_D$ is a $(k_1,k_2,\eps)$ {\em X-strong} extractor against $(b_1,b_2)$ (entangled) quantum storage provided
\begin{enumerate}[\indent(a)] \itemsep=-1pt
\item (entangled) $k_1+k_2-2b_2 \ge 2m+n-2+2\logeps$,
\item (non-entangled) $k_1+k_2-b_2 \ge 2m+n-2+2\logeps$.
\end{enumerate}
\end{lemma}
\begin{proof}
Recall that $E_D(x,y) = A_1x\cdot y, A_2x\cdot y,\ldots,A_mx\cdot y$ (see Def.~\ref{def:DEOR}). For
$0 \neq S \in \zo^m$, let $A_S = \sum_{i:S_i=1} A_i$ and note that $S\cdot E(x,y) = A_Sx\cdot y$. By the XOR-Lemma \ref{lem:xor},
\begin{align*}
& \norm{tr}{E(X,Y)\rho_{XY}X - U_m\rho_{XY}X} \le \sqrt{2^m\sum_{S \neq 0}\bignorm{tr}{(A_SX\cdot Y)\rho_{XY}X -
U_1\rho_{XY}X}^2}.
\end{align*}
The result then follows by Ineq. \eqref{pty:ext_cleve_x} in the proof of Cor.~\ref{cor:ext_ip} and its non-entangled
analogue.
\end{proof}

In a similar way, we also obtain a {\em Y-strong} extractor with analogous parameters. Following \cite{DEOR04}, we now
apply a seeded extractor against quantum storage (see Def.~\ref{def:seeded}) to the output of an X-strong (Y-strong)
extractor to obtain a two-source extractor with more output bits (see Lem.~\ref{lem:composition}).

\begin{definition}[\cite{TaShma09}]\label{def:seeded}
A function $E:\zo^n \times \zo^d \rightarrow \zo^m$ is a $(k,\eps)$ {\em seeded} extractor against $b$ quantum storage
if for any $n$-bit source $X$ with min-entropy $k$ and any $b$ qubit quantum storage $\rho_X$, \[
\norm{tr}{E(X,U_d)\rho_X - U_m\rho_X} \le \eps.
\]
\end{definition}

\begin{lemma}\label{lem:composition}
Let $E_B:\zo^n\times\zo^n \rightarrow \zo^d$ be a $(k_1,k_2,\eps)$ {\em X-strong} extractor against $(b_1,b_2)$
(entangled) quantum storage, and let $E_S:\zo^n\times\zo^d\rightarrow\zo^m$ and $E(x,y) = E_S(x, E_B(x,y))$.
\begin{enumerate}[\indent(a)] \itemsep=-1pt
\item (entangled) If $E_S$ is a $(k_1,\eps)$ {\em seeded extractor} against $b_1+b_2$ quantum storage then $E$ is a
$(k_1,k_2,2\eps)$ extractor against $(b_1,b_2)$ entangled quantum storage.
\item (non-entangled) If $E_S$ is a $(k_1,\eps)$ {\em seeded extractor} against $b_1$ quantum storage then $E$ is a
$(k_1,k_2,2\eps)$ extractor against $(b_1,b_2)$ non-entangled quantum storage.
\end{enumerate}
\end{lemma}

\begin{proof} Part (a):
$\norm{tr}{E_B(X,Y)\rho_{XY}X - U_d\rho_{XY}X} \le \eps$ and so $\norm{tr}{E_S(X,E_B(X,Y))\rho_{XY} -
E_S(X,U_d)\rho_{XY}} \le \eps$. But $\norm{tr}{E_S(X,U_d)\rho_{XY} - U_m\rho_{XY}} \le \eps$ by definition of $E_S$. The
result follows from the triangle inequality. For part (b) note that when the storage is non-entangled,
$\norm{tr}{E_S(X,U_d)\rho_X\rho_Y - U_m\rho_X\rho_Y} = \norm{tr}{E_S(X,U_d)\rho_X - U_m\rho_X}$, and it suffices
to require that $E_S$ be a seeded extractor against only $b_1$ quantum storage.
\end{proof}

A seeded extractor with almost optimal min-entropy loss is given in~\cite{DPVR09}. Their extractor is secure against
guessing entropy sources, and so trivially against quantum storage \cite{KT08} (see Sec.~\ref{sec:guessing} for details). We reformulate the seeded extractor in terms of Def.~\ref{def:seeded}.

\begin{corollary}[{\cite[Corrolary 5.3]{DPVR09}}]\label{cor:dpvr_seeded_ext}
There exists an explicit $(k,\eps)$ seeded extractor against $b$ quantum storage with seed length $d =
O(\log^3(n/\eps))$ and  $m = d+k-b-8\log(k-b)-8\logeps-O(1)$ output bits.
\end{corollary}

The proofs of Thms.~\ref{thm:deor_en} and \ref{thm:deor_unen} now follow by composing the explicit extractors of
Lem.~\ref{lem:deor_xstrong_en} and Cor.~\ref{cor:dpvr_seeded_ext} as in Lem.~\ref{lem:composition}.

\begin{proof-of-theorem}{\ref{thm:deor_en}}
$E_D$ is an X-strong extractor against entangled storage with $\half(k_1+k_2-2b_2-n-2\logeps)$ almost uniform output bits. This is larger than $O(\log^3(n/\eps))$ when $k_1+k_2-2b_2 > n+\Omega(\log^3(n/\eps))$, allowing us
to compose it with the seeded extractor secure against $b_1+b_2$ storage of Cor.~\ref{cor:dpvr_seeded_ext} on the
source $X$, obtaining $m = \half(k_1+k_2-2b_2-n-2\logeps) + (k_1-b_1-b_2) - 8\log(k_1-b_1-b_2) - 8\logeps - O(1)$. Similarly, $E_D$ is a Y-strong extractor, and can be composed with the seeded extractor on the source $Y$.
Choosing the better of the two, we prove the desired result.\footnote{We slightly sacrifice the parameters in
the formulation of the theorem to simplify the result.}
\end{proof-of-theorem}

\begin{proof-of-theorem}{\ref{thm:deor_unen}}
$E_D$ is an X-strong extractor against non-entangled storage with $\half(k_1+k_2-b_2-n-2\logeps)$ almost uniform output bits. This is larger than $O(\log^3(n/\eps))$ when $k_1+k_2-b_2 > n+\Omega(\log^3(n/\eps))$. Composing with the
seeded extractor secure against $b_1$ storage of Cor.~\ref{cor:dpvr_seeded_ext} on the source $X$ gives $m =
\half(k_1+k_2-b_2-n-2\logeps) + (k_1-b_1) - 8\log(k_1-b_1) - 8\logeps - O(1)$, and similarly for $Y$.
\end{proof-of-theorem}

\section{Guessing Entropy Adversaries}\label{sec:guessing}

In previous sections, we considered extractors in the presence of quantum adversaries with limited storage. A
stronger notion of quantum adversary was also studied in the literature~\cite{Ren05,KT08,FS08,DPVR09,TSSR10}.

\begin{definition}[\cite{KT08}] Let
$X\rho_X$ be an arbitrary cq-state. The guessing entropy of $X$ given $\rho_X$ is
\[ H_g(X \leftarrow \rho_X) := -\log \max_{M} \E_{x \leftarrow X}[\tr{M_x\rho_x}], \] where the maximum ranges over all
POVMs $M = \set{M_x}_{x \in X}$.
\end{definition}

Considering the probability distribution on the support of $X$ induced by measuring with $M$ on $\rho_X$ (which we
denote by $M(\rho_X)$), the above can be perhaps more easily understood as $H_g(X \leftarrow \rho_X) = -\log \max_{M} \Pr[M(\rho_X)=X]$.
Renner \cite{Ren05} considered sources with high {\em relative min-entropy}, rather than {\em guessing
entropy}. The two were shown to be equivalent \cite{KRS09}.

We can now define two-source extractors secure against non-entangled guessing entropy adversaries. Recall that in the
non-entangled case the bounded storage is given by $\rho_X \tensor \rho_Y$ (see Def.~\ref{def:storage}). Here, we place
a limit not on the amount of storage, but on the amount of information, in terms of guessing entropy, the adversaries
have on their respective sources. That is, we require that the guessing entropy of $X$ ($Y$) given $\rho_X$ ($\rho_Y$)
be high. We refer to the state $\rho_X \tensor \rho_Y$ as {\em quantum knowledge}, or if $\rho_x, \rho_y$ are classical
for every $x,y$, as {\em classical knowledge}.

\begin{definition}\label{def:ext_guessing}
A $(k_1, k_2, \eps)$ two-source extractor against quantum knowledge is a function $E:\zo^n \times \zo^n \rightarrow
\zo^m$ such that for any independent sources $X,Y$ and quantum knowledge $\rho_X\tensor\rho_Y$ with guessing
entropies $H_g(X \leftarrow \rho_X) \ge k_1$, $H_g(Y \leftarrow \rho_Y) \ge k_2$, we have
$\norm{tr}{E(X,Y)\rho_X\rho_Y - U_m\rho_X\rho_Y} \le \eps$.

The extractor is called {\em X-strong} if $\norm{tr}{E(X,Y)\rho_YX - U_m\rho_YX} \le \eps$. It is called {\em strong}
if it is both X-strong and Y-strong.
\end{definition}

It was shown that $H_g(X \leftarrow \rho_X) \ge H_\infty(X) - \log\dim(\rho_X)$ \cite{KT08}. Thus, we can view
adversaries with bounded quantum storage as a special case of general adversaries. In particular, a $(k_1-b_1, k_2-b_2,
\eps)$ extractor against quantum knowledge is trivially a $(k_1,k_2,\eps)$ extractor against {\em non-entangled}
$(b_1,b_2)$ storage.

\paragraph{One-bit output case:} K{\"o}nig and Terhal \cite{KT08} show that every classical one-bit output strong seeded extractor is also a strong extractor against quantum knowledge with roughly the same parameters. They reduce
the "quantum security" of the extractor to the "classical security", {\em irrespective} of the entropy of the seed. Informally, $\norm{tr}{E(X,Y)\rho_XY - U_1\rho_XY}$ is small if the statement is also true when $\rho_X$ is classical.
We give a version of their Lem. 2 with slightly improved parameters. The lemma shows that it suffices to prove security of an extractor with respect only to classical knowledge obtained by performing a Pretty Good Measurement (PGM) \cite{HW94} on arbitrary quantum knowledge. For a cq-state $Z\rho_Z$, a PGM is a POVM $\cE=\set{\cE_z}_{z\in Z}$
such that $\cE_z = p(z)\rho_Z^{-1/2}\rho_z\rho_Z^{-1/2}$.

\begin{lemma}\label{lem:kt}
Let $Z\rho_Z$ be a cq-state, and $f$ be a Boolean function. Then\footnote{$\cE(\rho_Z)$ is a classical probability
distribution and the trace distance $\norm{tr}{f(Z)\cE(\rho_Z) - U\cE(\rho_Z)}$ reduces to the classical variational
distance.} \[ \norm{tr}{f(Z)\rho_Z - U\rho_Z} \le \sqrt{\half\norm{tr}{f(Z)\cE(\rho_Z) - U\cE(\rho_Z)}}, \] where
$\cE=\set{\cE_z}_{z\in Z}$ is a Pretty Good Measurement, $\cE_z = p(z)\rho_Z^{-1/2}\rho_z\rho_Z^{-1/2}$.
\end{lemma}
\begin{proof}
We need the following lemma.

\begin{lemma}[{\cite[Lemma 5.1.3]{Ren05}}]\label{lem:ren_cond_l2}
Let $S$ be a Hermitian operator and let $\sig$ be a nonnegative operator. Then $\norm{tr}{S} \le
\half\sqrt{\tr{\sig}\tr{\sig^{-1/2}S\sig^{-1/2}S}}.$
\end{lemma}

Denote $\rho = \rho_Z$, $\rho_b = \sum_{z:f(z)=b} p(z)\rho_z$ for $b=0,1$. Further define (informally) a POVM $M$ for
guessing $f$ from $\rho_Z$ by first applying $\cE$ to get $z$ and then computing $f(z)$. Then
\begin{align*}
\Pr[M(\rho_Z) = f(Z)] & = \sum_z p(z)\sum_{z':f(z')=f(z)}\tr{\cE_{z'}\rho_z} \\
& = \tr{\sum_{f(z')=f(z)}\rho^{-1/2}(p(z')\rho_{z'})\rho^{-1/2}(p(z)\rho_z)} \\
& = \tr{\rho^{-1/2}\rho_0\rho^{-1/2}\rho_0 + \rho^{-1/2}\rho_1\rho^{-1/2}\rho_1},
\end{align*}
and similarly $\Pr[M(\rho_Z) \neq f(Z)] = \tr{\rho^{-1/2}\rho_0\rho^{-1/2}\rho_1 +
\rho^{-1/2}\rho_1\rho^{-1/2}\rho_0}$. Hence \begin{align}
\abs{\Pr[M(\rho_Z) = f(Z)] - \Pr[M(\rho_Z) \neq f(Z)]} =  \tr{\rho^{-1/2}(\rho_0-\rho_1)\rho^{-1/2}(\rho_0-\rho_1)}. \label{eq:kt_lem2_1}
\end{align}
By Eq. \eqref{eq:trace_one_bit}, $\norm{tr}{f(Z)\rho_Z - U\rho_Z} = \norm{tr}{\rho_0 - \rho_1}$, and by
Lem.~\ref{lem:ren_cond_l2}, setting $S = \rho_0 - \rho_1$, $\sig = \rho$, \begin{align} \norm{tr}{\rho_0 - \rho_1} \le
\half\sqrt{\tr{\rho^{-1/2}(\rho_0-\rho_1)\rho^{-1/2}(\rho_0-\rho_1)}}. \label{eq:kt_lem2_2}
\end{align}
Combining Eq. \eqref{eq:kt_lem2_1} with Eq. \eqref{eq:kt_lem2_2} gives \[
\norm{tr}{f(Z)\rho_Z - U\rho_Z} \le \sqrt{\frac{1}{4}\abs{\Pr[M(\rho_Z) = f(Z)] - \Pr[M(\rho_Z) \neq f(Z)]}}.
\]
Finally, \[
\abs{\Pr[M(\rho_Z) = f(Z)] - \Pr[M(\rho_Z) \neq f(Z)]} \le 2\norm{tr}{f(Z)M(\rho_Z) - UM(\rho_Z)} \le
2\norm{tr}{f(Z)\cE(\rho_Z) - U\cE(\rho_Z)}, \]
as the left hand side describes a trivial strategy to guess $f$ from $M(\rho)$, giving the desired
result.
\end{proof}

\begin{corollary}
If $E$ is a classical one-bit output $(k_1,k_2,\eps)$ two-source extractor, then it is a
$(k_1+\logeps,k_2+\logeps,\sqrt{3\eps/2})$ two-source extractor against quantum knowledge.
\end{corollary}
\begin{proof}
By Lem.~\ref{lem:kt}, $\norm{tr}{E(X,Y)\rho_X\rho_Y - U\rho_X\rho_Y} \le \sqrt{\half\norm{tr}{E(X,Y)\cE(\rho_X\rho_Y) -
U\cE(\rho_X\rho_Y)}}$. A direct calculation shows that for every $x,y$, $\cE(\rho_x\tensor\rho_y) =
\cE_1(\rho_x)\tensor\cE_2(\rho_y)$, where $\cE_1,\cE_2$ are Pretty Good Measurements on states $X\rho_X, Y\rho_Y$
respectively. In other words, $\cE(\rho_X \otimes\rho_Y)$ induces a classical distribution $C_X\tensor C_Y$. Thus
\begin{align} \norm{tr}{E(X,Y)\rho_X\rho_Y - U\rho_X\rho_Y} \le \sqrt{\half\norm{tr}{E(X,Y)C_XC_Y - UC_XC_Y}},
\label{eq:ext_q_to_c}
\end{align}
where $H_g(X \leftarrow C_X) \ge H_g(X \leftarrow \rho_X)$, and the same for $Y$.

By the definition of (classical) guessing entropy, one can easily show that a classical $(k_1,k_2,\eps)$ two-source
extractor is a $(k_1+\logeps, k_2+\logeps, 3\eps)$ extractor against {\em classical knowledge} (for details see
Proposition 1 in \cite{KT08}). Ineq. \eqref{eq:ext_q_to_c} then gives the desired parameters against quantum knowledge.
\end{proof}

By a similar argument and following the proof of Theorem 1 in $\cite{KT08}$, we get

\begin{corollary}\label{cor:kt_one_bit_secure}
If $E$ is a classical one-bit output $(k_1,k_2,\eps)$ X-strong extractor, then it is a $(k_1,k_2+\logeps,\sqrt{\eps})$
X-strong extractor against quantum knowledge.
\end{corollary}

\paragraph{The multi-bit output case:} We now show how to apply the results in the one-bit case, together with our
XOR-Lemma~\ref{lem:xor}, to show security in the multi-bit case, proving Thm.~\ref{thm:deor_knowledge}.

By Ineq.~\eqref{pty:ext_cleve_x} in the proof of Cor.~\ref{cor:ext_ip}, inner product is a {\em classical} X-strong
extractor with error $\eps \leq 2^{-(k_1+k_2-n+2)/2}$. Plugging this into Cor.~\ref{cor:kt_one_bit_secure} we obtain

\begin{corollary}\label{cor:ext_ip_knowledge}
The function $E_{{IP}_A}(x,y)=Ax \cdot y$, for any full rank matrix $A$, is a $(k_1,k_2,\eps)$ X-strong (Y-strong)
extractor against quantum knowledge provided that $k_1 + k_2 \ge n - 2 + 6\logeps$.
\end{corollary}

We now repeat the steps performed in Sec.~\ref{sec:many} in the setting of non-entangled guessing entropy adversaries
to obtain a multi-bit extractor against quantum knowledge. In exactly the same fashion as in the proof of
Lem.~\ref{lem:deor_xstrong_en} we use the XOR-Lemma~\ref{lem:xor} to reduce the security of $E_D$ to the strong one-bit
case of Cor.~\ref{cor:ext_ip_knowledge}.

\begin{lemma}\label{lem:deor_xstrong_knowledge}
$E_{D}$ is a $(k_1,k_2,\eps)$ X-strong (Y-strong) extractor against quantum knowledge provided that $k_1 + k_2 \ge 6m +
n - 2 + 6\logeps$.
\end{lemma}
\begin{proof}
By the XOR-Lemma~\ref{lem:xor} and Cor.~\ref{cor:ext_ip_knowledge},
\begin{align*}
\norm{tr}{E(X,Y)\rho_YX - U_m\rho_YX} \le \sqrt{2^m\sum_{S \neq 0}\bignorm{tr}{(A_SX\cdot Y)\rho_YX -
U_1\rho_YX}^2} \le 2^m\cdot2^{-(k_1+k_2-n+2)/6}.
\end{align*}
\end{proof}

To obtain our final result, we now compose our strong extractor with a seeded extractor against quantum knowledge.

\begin{lemma}
Let $E_B:\zo^n\times\zo^n \rightarrow \zo^d$ be a $(k_1,k_2,\eps)$ {\em X-strong} extractor against quantum knowledge
and let $E_S:\zo^n\times\zo^d\rightarrow\zo^m$ be a $(k_1,\eps)$ {\em seeded extractor} against quantum
knowledge\footnote{For a formal definition see~\cite{DPVR09}.}. Then $E(x,y) = E_S(x,
E_B(x,y))$ is a $(k_1,k_2,2\eps)$ extractor against quantum knowledge.
\end{lemma}
\begin{proof}
Immediate from the extractor definitions and the triangle inequality.
\end{proof}

\begin{corollary}[{\cite[Corrolary 5.3]{DPVR09}}]\label{cor:dpvr_ext_knowledge}
There exists an explicit $(k,\eps)$ seeded extractor against quantum knowledge with seed length $d =
O(\log^3(n/\eps))$ and $m = d+k-8\log k-8\logeps-O(1)$.
\end{corollary}

\begin{proof-of-theorem}{\ref{thm:deor_knowledge}}
$E_D$ is an X-strong extractor against quantum knowledge with $\frac{1}{6}(k_1 + k_2 - n - 6\logeps) - O(1)$
output bits.  This is larger than $O(\log^3(n/\eps))$ when $k_1+k_2 > n+\Omega(\log^3(n/\eps))$. Composing with
the seeded extractor of Cor.~\ref{cor:dpvr_ext_knowledge} on the source $X$ gives
$m = \frac{1}{6}(k_1+k_2-n-6\logeps) + k_1 - 8\log k_1 - 8\logeps - O(1)$, and similarly for $Y$.
\end{proof-of-theorem}

\section*{Acknowledgments}
The authors would like to thank Nir Bitansky, Ashwin Nayak, Oded Regev, Amnon Ta-Shma, Thomas Vidick and Ronald de Wolf
for valuable discussions. We are especially indebted to Ronald de Wolf for allowing us to use his exact protocol for IP
in the SMP model with entanglement, and to Thomas Vidick for pointing out how to replace $D$ with $M$ in our XOR-Lemma,
which allowed us to prove Thm.~\ref{thm:deor_knowledge}.

\newcommand{\etalchar}[1]{$^{#1}$}

%\bibliographystyle{alphaabbrv}
%\bibliography{two-source}

\appendix

\section{Many Bit Extractors Against Quantum Storage from Classical Storage}\label{app:storage}

K{\"o}nig and Terhal \cite{KT08} prove that any (classical) seeded extractor is secure against {\em non-entangled}
quantum storage, albeit with exponentially larger (in the storage size) error. Their proof is also valid for X-strong
(Y-strong) two-source extractors.

Their Lemma 5 essentially shows that every $(k_1,k_2,\eps)$ X-strong extractor has error $4\cdot 2^{3b_2}\cdot\eps$
against $(b_1,b_2)$ quantum storage (for any $b_1$), assuming $H_\infty(X) \ge k_1$ and $H_g(Y \leftarrow \rho_Y) \ge
k_2 + \logeps$. Recall that $H_g(Y \leftarrow \rho_Y) \ge H_\infty(Y) - b_2$. Adapted to our definitions, their result
is

\begin{lemma}[{\cite[Lemma 5]{KT08}}]\label{lem:deor_xstrong_kt}
Let $E$ be a $(k_1,k_2,\eps)$ X-strong extractor. Then $E$ is a $(k_1, k_2 + b_2 + \logeps, 4\cdot2^{3b_2}\eps)$
X-strong extractor against $(b_1,b_2)$ non-entangled storage.
\end{lemma}

In particular, this shows that $E_D$ is an X-strong extractor with $m = k_1 + k_2 - 10b_2 - n - 4 - 3\logeps$. For
comparison, our Lem.~\ref{lem:deor_xstrong_en} gives $m = \half(k_1 + k_2 - b_2 - n + 2 - 2\logeps)$, which is better
when the storage is large, say, $b_2 \ge k_2/19$.

For completeness, we derive an alternate version of Thm.~\ref{thm:deor_unen} based on Lem.~\ref{lem:deor_xstrong_kt},
by composing the extractor above with the seeded extractor of \cite{DPVR09}.

\begin{theorem}
The DEOR-construction is a $(k_1,k_2,\eps)$ extractor against $(b_1,b_2)$ non-entangled storage with $m =
(1-o(1))\max(k_1-9b_2, k_2-9b_1) + k_1-b_1 + k_2-b_2 - n - 11\logeps - O(1)$ output bits provided $k_1 + k_2 - 10\max(b_1,b_2) > n + \Omega(\log^3(n/\eps))$.
\end{theorem}

Here too we are able to extract more bits than guaranteed by Thm.~\ref{thm:deor_unen} when the storage is symmetric and
constitutes a small fraction $(<1/19)$ of the min-entropy. In particular, the storage must be at least ten times
smaller than the min-entropy, whereas no such restriction exist in Thm.~\ref{thm:deor_unen}.

We note that it is not immediately possible to obtain an analogue of Lem.~\ref{lem:deor_xstrong_kt} for weak two-source
extractors. The proof relates the security of an extractor with respect to quantum side information, to its security
with respect to classical side information. In the weak extractor setting, it thus suffices to consider classical side
information of the form $\cF(\rho_X\tensor\rho_Y)$ for some specific POVM $\cF$ given in the proof. The problem with
this approach is that generally $\cF(\rho_X\tensor\rho_Y)$ might induce a random variable $C_{XY}$ correlated with both
$X$ and $Y$, breaking the independence assumption (i.e., when conditioning on values of $C_{XY}$, $X$ and $Y$ might not
be independent) and rendering the classical extractor insecure. It is not inconceivable that $\cF$ does have the
property $\cF(\rho_X\tensor\rho_Y) = C_X\tensor C_Y$, but we leave this open.

\end{document}